\documentclass[11pt]{article}

\usepackage{mathptmx}

\usepackage{setspace}
\usepackage{algorithm}
\usepackage{algorithmic}
\usepackage[utf8]{inputenc} % allow utf-8 input
\usepackage[T1]{fontenc}    % use 8-bit T1 fonts
\usepackage{hyperref}       % hyperlinks
\usepackage{url}            % simple URL typesetting
\usepackage{booktabs}       % professional-quality tables
\usepackage{tabularx}
\usepackage{amsthm}
\usepackage{amsmath}
\usepackage{amsfonts}       % blackboard math symbols
\usepackage{nicefrac}       % compact symbols for 1/2, etc.
\usepackage{microtype}      % microtypography
\usepackage{lipsum}
\usepackage{fancyhdr}       % header
\usepackage{graphicx}       % graphics
\graphicspath{{media/}}     % organize your images and other figures under media/ folder
\usepackage{caption}
\usepackage[authoryear]{natbib}
\graphicspath{ ./ }
\DeclareMathOperator*{\argmin}{arg\,min}

\usepackage{enumitem}

\usepackage[most]{tcolorbox}

\newtcolorbox[auto counter, number within=section]{example}[2][]{%
  breakable,
  fonttitle=\bfseries,
  title=Example~\thetcbcounter: #2,#1
}

\usepackage{geometry}
\usepackage{array}
\newtheorem*{theorem*}{Theorem}

\title{Finding the Center and Centroid of a Graph with Multiple Sources
}

\author{
  Matthew Chou \\
}

\begin{document}
\maketitle
\onehalfspacing

\begin{abstract}
We consider the problem of finding a "fair" meeting place when S people want to get together. Specifically, we will consider the cases where a "fair" meeting place is defined to be either 1) a node on a graph that minimizes the maximum time/distance to each person or 2) a node on a graph that minimizes the sum of times/distances to each of the sources. In graph theory, these nodes are denoted as the center and centroid of a graph respectively. In this paper, we propose a novel solution for finding the center and centroid of a graph by using a multiple source alternating Dijkstra's Algorithm. Additionally, we introduce a stopping condition that significantly saves on time complexity without compromising the accuracy of the solution. The results of this paper are a low complexity algorithm that is optimal in computing the center of $S$ sources among $N$ nodes and a low complexity algorithm that is close to optimal for computing the centroid of $S$ sources among $N$ nodes.
\\
\end{abstract}

\section{Introduction}
\label{sec:Introduction}
There exist many real world problems where finding the center of a graph is deemed useful where the center can be defined in many different ways. As an example, in the field of computer networks, one is often interested in finding the center of nodes in a network to determine where to place servers \citep{serverproblem}. Another example is the facility location problem \citep{facilityproblem} where one wants to build a set of facilities that lies within center of a set of nodes where each node may represent residential homes, businesses, etc. In this paper, we want to examine a problem where several people are interested in meeting at a common location. This location should be "fair" in the sense that it either minimizes the maximum time/distance for all of the people or it minimizes the total time/distance for all of the people. Solving this problem amounts to finding the center or centroid of the nodes that represent the people in the graph. What distinguishes this problem from the aforementioned problems is that the previous problems are solvable offline whereas in our scenario, there may be a need to solve it continuously as the ideal location of the center may vary with time. As a result we need an efficient solution. Furthermore, in our problem we are trying to find the center for a few nodes that exist among many nodes whereas the previous problems are trying to find the center among many nodes. To the best of our knowledge, our problem has not been widely explored. 

The above problems as well as our problems are equivalent to finding the center of nodes in a graph. In the facility location problem \citep{facilityproblem} and network server problem \citep{serverproblem}, there may be N nodes and one wants to find the center of these N nodes. In our problem, there are S nodes among N total nodes, where $S << N$, and we want to find the center of the S nodes among the N total nodes. The former problem is widely known and studied under the general topic of the K-Center problem \citep{Kcenter}, K-Median problem \citep{Kmedian}, and the Jordan center problem \citep{Jordancenter}. Our problem can be cast as the Jordan center problem except the center is to be found among S nodes instead of N nodes. A well known solution to the Jordan center problem is the Floyd-Warshall algorithm \citep{floydWarshall}. There have been many other theoretical works that have studied solutions for these problems \citep{solutions}.

The contributions of this paper are as follows 
\begin{itemize}
    \item We construct the problem of finding the center and centroid of S sources among N nodes where $S<<N$
    \item We propose a solution using a multiple source Dijkstra's Algorithm or A* Algorithm
    \item We propose a stopping condition for our algorithm that is optimal for finding the center while providing significant complexity savings
    \item We propose a modified stopping condition based on the stopping condition that is used for finding the center to find the centroid. The modified stopping condition may be suboptimal for finding the centroid, and we measured through simulation the amount of degradation in addition to the amount of complexity savings.
\end{itemize}
We will start by providing a background followed by describing our algorithmic framework and then describe several methods that improve upon the efficiency of our algorithmic framework.

\section{Background}
\label{sec:Background}
In this paper we define the problem of finding the center of S sources among N nodes as the S-source-center problem. The solution provided in this paper requires the definition of the center as well as an explanation of Dijkstra's Algorithm and  A* Algorithm. 

\subsection{Center and Centroid Definition}
To define the center of the S-source-center problem, we must first discuss the center in a broader context. The \textit{eccentricity} of a node is defined as the maximum distance from a node to any other node within the graph and the center is defined as node(s) with minimum eccentricity. We define the eccentricity of a node of the S-source-center problem as the maximum distance from the node to any of the S source nodes and the minimum of this eccentricity as the center. In addition to the center, we also need to define the centroid of a graph. The centroid of a graph is defined as the node with the minimum sum of distances to all other nodes within a graph. In the context of the S-source-center problem, the centroid of a graph is the set of nodes with the minimum sum of distances to all of the source nodes within a graph. The difference between the center and the centroid is apparent in the weighting of the nodes. Whereas the center will equalize the distances to each of the source nodes, the centroid gives more weight to clusters of nodes as opposed to singular nodes. In practical scenarios, the preference of finding either the center or centroid for the S-source-center problem will likely be subjective.

\subsection{Dijkstra's Algorithm}
Dijkstra's Algorithm \citep{dijkstra} is a graph exploration algorithm that starts from a source node and explores outwards, keeping track of the shortest distance to each node. Each node in the graph is visited only once and the output of the algorithm is a single sourced shortest path tree which is commonly used for mapping and network routing applications. The algorithm maintains a priority queue in which nodes are extracted as they get marked as finished. The priority queue starts at the source node and adds all neighboring nodes into the queue. After every iteration, all neighboring nodes of the current node have its distance computed and updated if a smaller distance is found. The node that is extracted from the list is the node with the minimum distance value within the priority queue.

A summarized version of the algorithm is provided for reference (see Algorithm~\ref{alg:dijkstras}).

%Algorithm 1
\begin{algorithm}
\caption{Dijkstra's Algorithm}
\label{alg:dijkstras}

1) Initialize the distance as either 0 for the source node or $\infty$ for all other nodes

2) Extract the node from the priority queue with the smallest distance. Call this node the current node

3) For each of the current node's unvisited neighboring nodes, calculate the distance from the current node and update it if the new distance is smaller

4) Mark the current node as visited and extract it from the priority queue. A visited node cannot be revisited. Repeat steps 2-4 until all nodes have been visited

\end{algorithm}

\subsection{A* Algorithm}
\label{sec:astar}
The A* algorithm \citep{Astar} is a graph exploration algorithm that searches for the shortest path to a destination. It utilizes Dijkstra's Algorithm with a modified cost function to create a more efficient algorithm for finding a path to a single point. The cost function for the A* algorithm is $f(n) = g(n) + h(n)$ where $g(n)$ is the actual cost from the source node and $h(n)$ is a heuristic function that represents an estimated cost to the destination node. The purpose of adding a heuristic function $h(n)$ to the cost function is to help steer the order of evaluation of nodes towards nodes that are in the direction of the destination. The A* algorithm also utilizes a priority queue that is based on the f(n) value and extracts the node with the smallest $f(n)$ value. There are many different heuristic functions that can be used with the algorithm and depending on the heuristic function used, one can get different results. Common heuristic functions are Euclidean distance, Manhattan distance, etc.

\section{Previous Works}
\label{sec:prevWorks}
The previous problems of the K-Center \citep{Kcenter}, K-Median \citep{Kmedian}, and Jordan Center \citep{Jordancenter} all have solutions \citep{solutions} which are different than ours. First, the K-Center problem utilizes a selecting algorithm that given a set of data, selects the best points to create facilities either at random, through 2 approximation method, or other means. Like the K-Center problem, the K-Median problem utilizes a selecting algorithm that works either randomly or updates itself upon every iteration. Both of these solutions do not have any practical application to the S-source-center problem we proposed. The Jordan center, however, is very similar to the S-source-center problem, with the only difference being how many nodes the problem considers for finding the center. The Jordan center requires the consideration of all nodes whilst the S-source-center problem only requires the consideration of the S source nodes. The common approach to the Jordan center problem is the Floyd-Warshall \citep{floydWarshall} algorithm which solves the problem in O($V^3$) time. On the other hand, our algorithm is comprised of S alternating Dijkstra's Algorithms from each of the S sources. Each Dijkstra's Algorithm has time complexity of $O((V+E)log(V))$ given a binary min-heap implementation \citep{fredman1987fibonacci}. Therefore, by alternating S Dijkstra's Algorithms, the time complexity is $O(S(V+E)log(V))$. Furthermore, we propose a stopping condition so that not all vertices, V, need to be explored in each of the S Dijkstra's Algorithm which leads to significant time complexity savings without compromising accuracy.

\section{Algorithmic Framework}
\label{sec:algorithmFramework}
The two different problems of finding the centroid and center of the S-source-center problem have different objectives requiring different algorithms for each one. In the following subsections, we will present both problems as well as a solution and an example for each one. In our algorithms, the reference to distance in a graph is general and may represent either distance, time, or some other measure. Furthermore, our algorithm works for both undirected and directed graphs with positive weights.

\subsection{Center}
\label{sec:Center}
To find the center of the S-source-center problem we have to solve the following optimization problem:
\begin{equation}
\hat{v} = \argmin_{v_i \in V} \max_{s_j \in S} d(v_i, s_j)
\end{equation}
where V is the set of all vertices within the graph and S is the set of all source nodes.

 We want to determine the node $\hat{v}$ with the minimum value of f($v_j) = \max_{s_j \in S} d(v_i, s_j)$. The first proposed solution is running a Dijkstra's Algorithm in alternating steps from each source node and storing the shortest path for every node. By storing the shortest path from every node to every source node, one is able to determine the f($v_j$) value for all nodes by keeping track of the maximum shortest path to all source nodes. The algorithm will also maintain a separate priority queue for each source node and alternate between the queues after a node has been extracted from a queue. This is done so that the algorithm switches between source nodes while exploring. Similar to our algorithm is the Bidirectional Dijkstra's Algorithm \citep{bidijkstra}, but unlike Bidirectional Dijkstra's Algorithm, our algorithm it is used to find the shortest path between two nodes. Our proposed algorithm is shown in Algorithm~\ref{alg:centerAlgorithm}).

\begin{algorithm}
\caption{Multiple Source Dijkstra's Algorithm for Finding the Center}
\label{alg:centerAlgorithm}
1) Initialization: Create a priority queue for each source node and initialize the distance of source nodes to 0 and all the other nodes to $\infty$

2) Selection: Pick the unvisited node with the smallest d($s_i$,v) (initially the source node) and extract node

3) Relaxation: For the extracted node, check all neighboring nodes and compare/update its distances. If extracted node has been visited by all sources, keep track of its maximum distance to all source nodes

4) Alternation: The extracted node is marked as visited and will not be visited again. Repeat steps 2-3 alternating between each of the sources, for all nodes

5) Stopping: Find the node with the minimum eccentricity to all nodes. Alternatively, the minimum eccentricity could be kept track of whenever a node is extracted and has been visited by all sources.\\
\end{algorithm}

%Tables
\begin{figure}
\begin{center}
\includegraphics[scale=0.25]{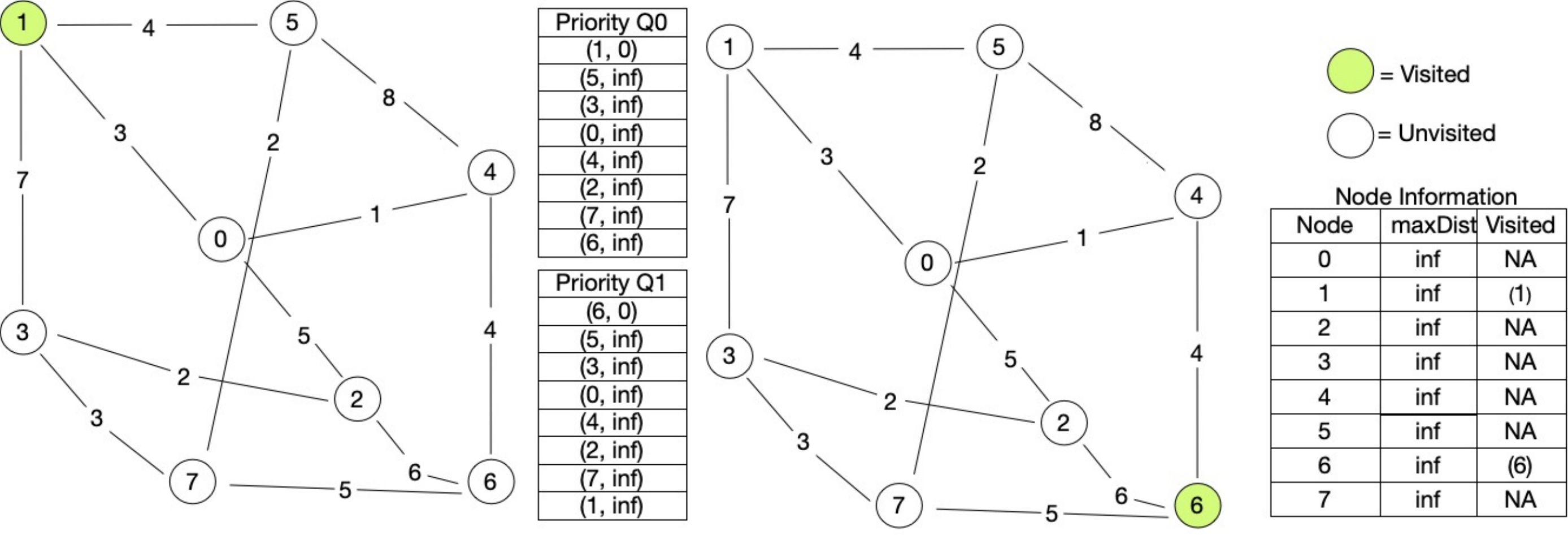}
\caption{The state of our proposed algorithm at the start of the first iteration. During this iteration, Node 1 is extracted from priority queue 0, and all it neighboring nodes' distances are calculated and updated. Similarly, Node 6 is extracted from priority queue 1 and all its neighboring nodes' distances are calculated and updated}
\label{fig:figure_1}
\end{center}
\end{figure}

\begin{figure}
\begin{center}
\includegraphics[scale=0.25]{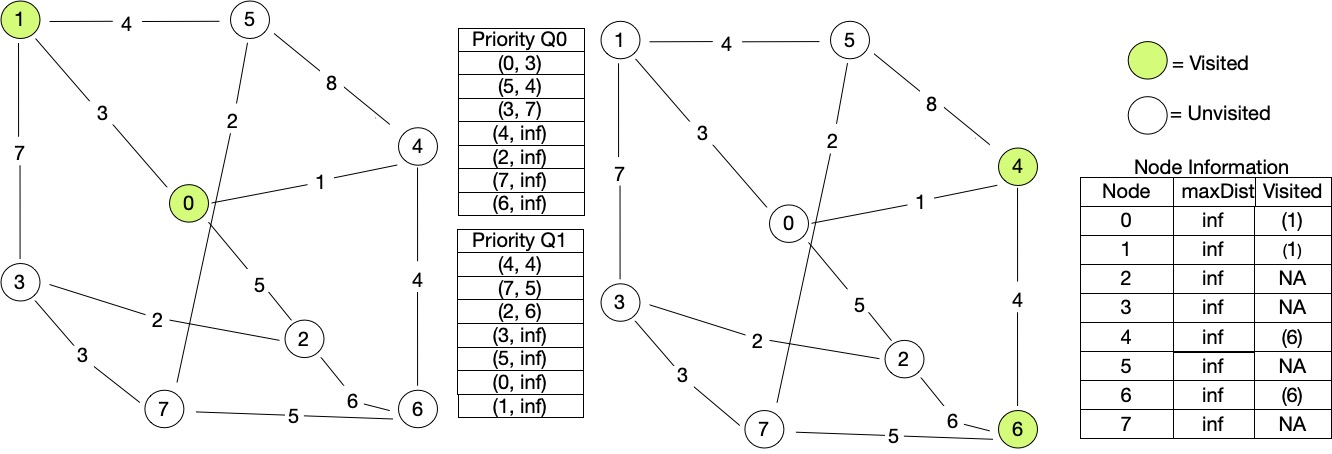}
\caption{The state of our proposed algorithm at the start of the second iteration. During this iteration, Node 0 is extracted from priority queue 0, and all it neighboring nodes' distances are calculated and updated. Similarly, Node 4 is extracted from priority queue 1 and all its neighboring nodes' distances are calculated and updated}
\label{fig:figure_2}
\end{center}
\end{figure}

\begin{figure}
\begin{center}
\includegraphics[scale=0.25]{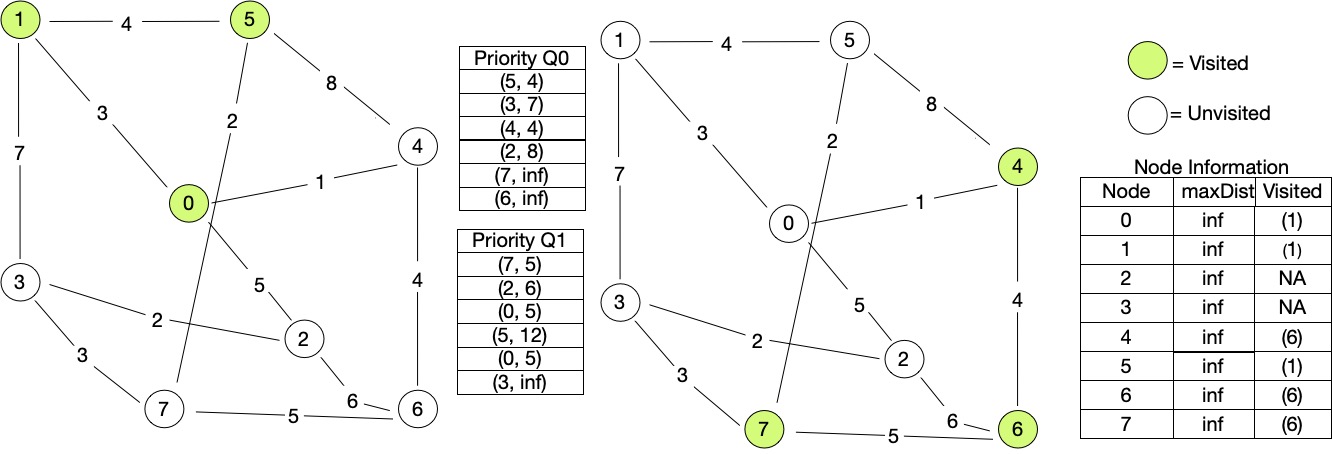}
\caption{The state of our proposed algorithm at the start of the third iteration. During this iteration, Node 5 is extracted from priority queue 0, and all it neighboring nodes' distances are calculated and updated. Similarly, Node 7 is extracted from priority queue 1 and all its neighboring nodes' distances are calculated and updated}
\label{fig:figure_3}
\end{center}
\end{figure}

\begin{figure}
\begin{center}
\includegraphics[scale=0.25]{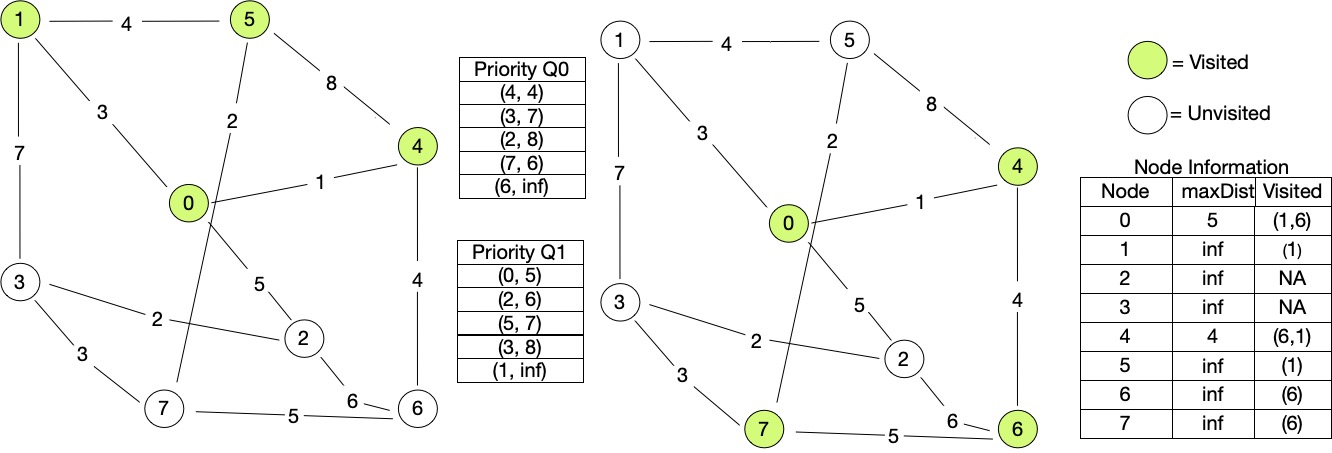}
\caption{The state of our proposed algorithm at the start of the fourth iteration. During this iteration, Node 4 is extracted from priority queue 0, and all it neighboring nodes' distances are calculated and updated. Similarly, Node 0 is extracted from priority queue 1 and all its neighboring nodes' distances are calculated and updated}
\label{fig:figure_4}
\end{center}
\end{figure}

\begin{figure}
\begin{center}
\includegraphics[scale=0.25]{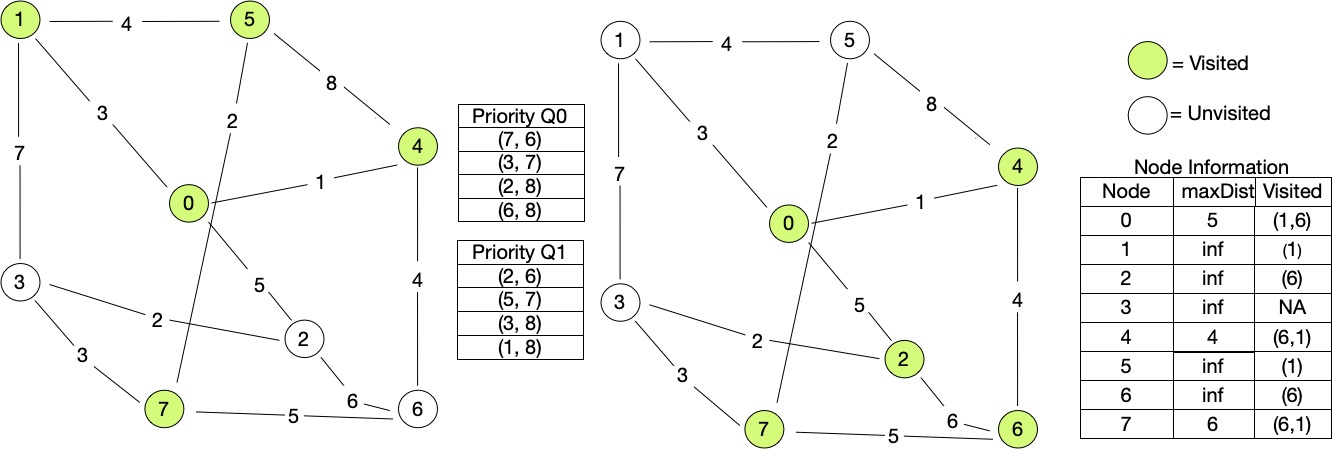}
\caption{The state of our proposed algorithm at the start of the fifth iteration. During this iteration, Node 7 is extracted from priority queue 0, and all it neighboring nodes' distances are calculated and updated. Similarly, Node 2 is extracted from priority queue 1 and all its neighboring nodes' distances are calculated and updated}
\label{fig:figure_5}
\end{center}
\end{figure}

\begin{figure}
\begin{center}
\includegraphics[scale=0.25]{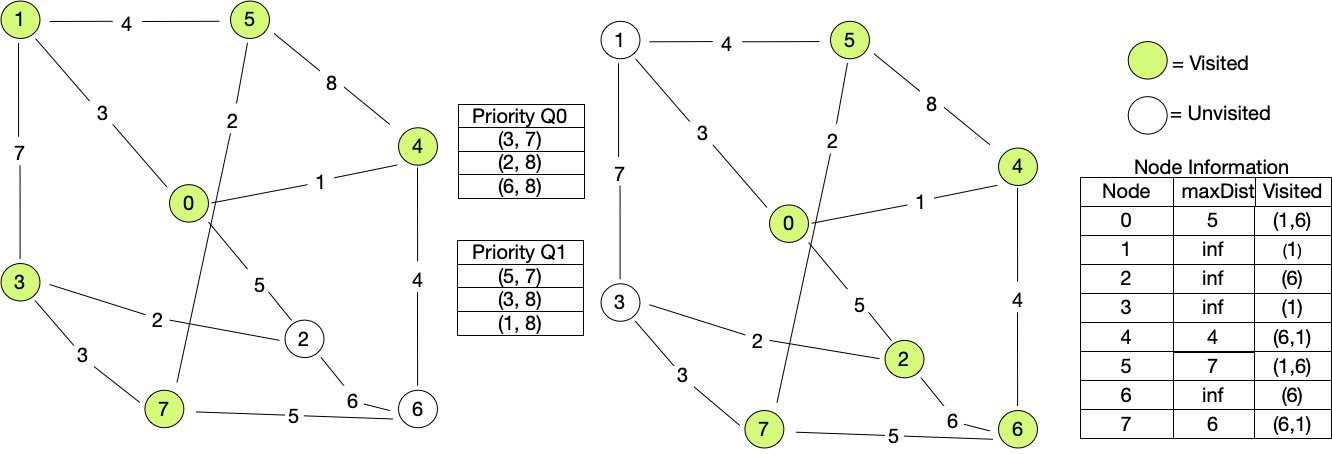}
\caption{The state of our proposed algorithm at the start of the sixth iteration. During this iteration, Node 3 is extracted from priority queue 0, and all it neighboring nodes' distances are calculated and updated. Similarly, Node 5 is extracted from priority queue 1 and all its neighboring nodes' distances are calculated and updated}
\label{fig:figure_6}
\end{center}
\end{figure}

\begin{figure}
\begin{center}
\includegraphics[scale=0.25]{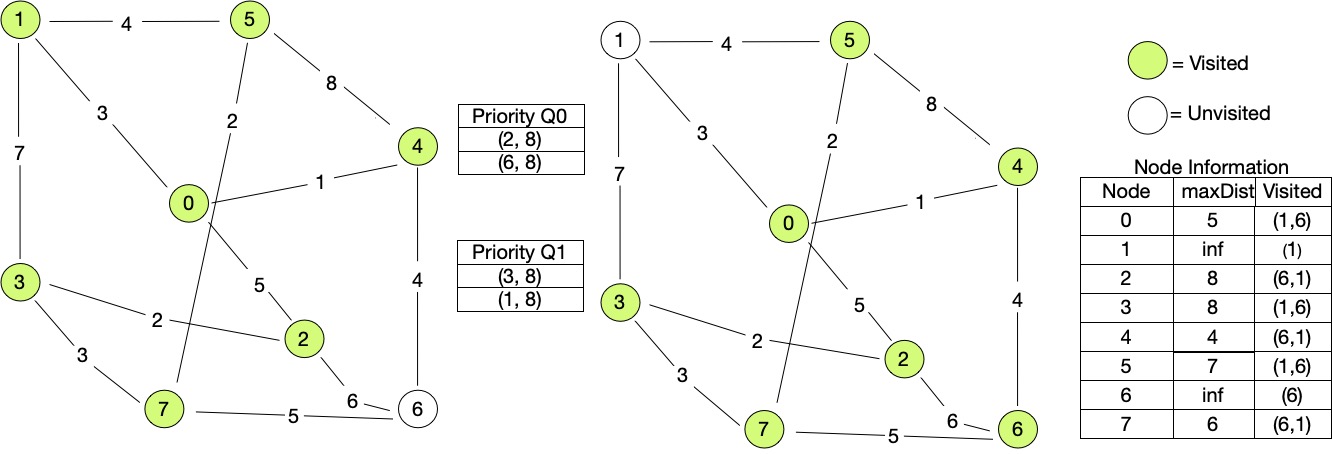}
\caption{The state of our proposed algorithm at the start of the seventh iteration. During this iteration, Node 2 is extracted from priority queue 0, and all it neighboring nodes' distances are calculated and updated. Similarly, Node 3 is extracted from priority queue 1 and all its neighboring nodes' distances are calculated and updated}
\label{fig:figure_7}
\end{center}
\end{figure}

\begin{figure}
\begin{center}
\includegraphics[scale=0.25]{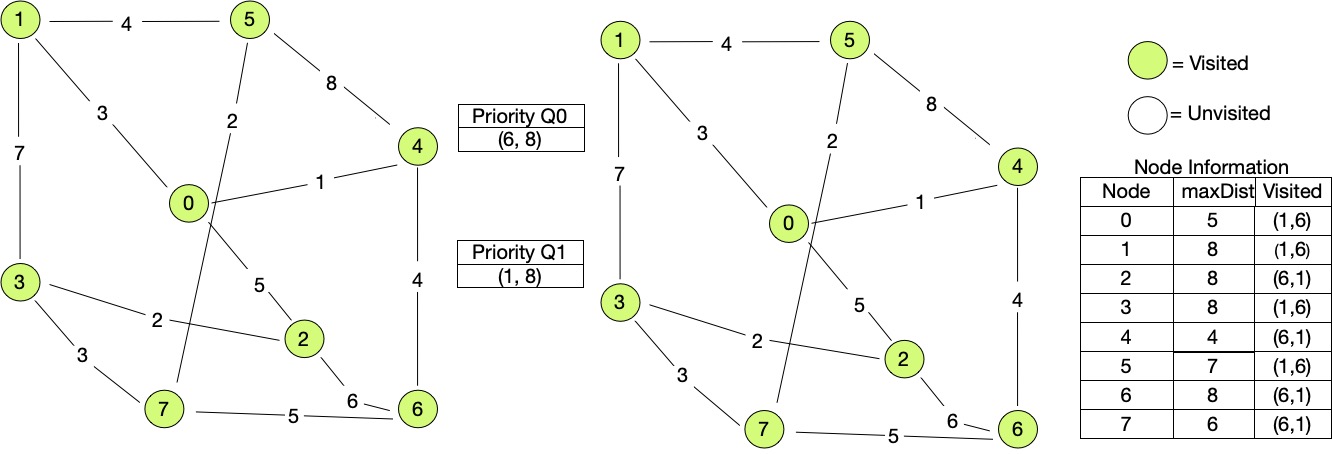}
\caption{The state of our proposed algorithm at the start of the eigthth iteration. During this iteration, Node 6 is extracted from priority queue 0, and all it neighboring nodes' distances are calculated and updated. Similarly, Node 1 is extracted from priority queue 1 and all its neighboring nodes' distances are calculated and updated}
\label{fig:figure_8}
\end{center}
\end{figure}

\begin{figure}
\begin{center}
\includegraphics[scale=0.5]{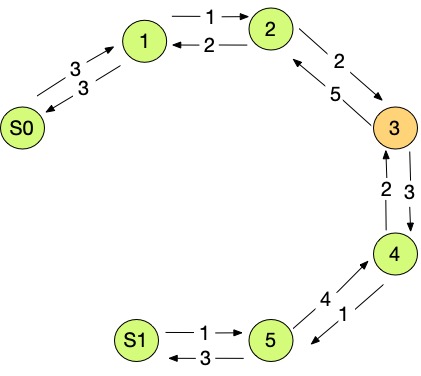}
\caption{A case in which the heuristic for finding the centroid fails to find the optimal solution}
\label{fig:figure_9}
\end{center}
\end{figure}

To illustrate how the algorithm works, we provide the following example (see Example~\ref{ex:example1}). In the example we have two source nodes, Nodes 1 and 6. The graph is both undirected and contains all positive, but not identical weights and has 8 total vertices and 12 total edges.

%Example 1
\begin{example}[label={ex:example1}]{Finding the Center with No Stopping Condition}
\textbf{Initialization:}\\
Create a priority queue for both sources:\\
\hspace*{2em} a) Priority Queue 0: Node 1 distance is set to 0 and all others are set to $\infty$\\
\hspace*{2em} b) Priority Queue 1: Node 6 distance is set to 0 and all others are set to $\infty$\\
Then, we perform steps 2 and 3 for each source in alternating fashion for all 8 nodes in the graph\\

\textbf{Iteration 1:}(see Figure~\ref{fig:figure_1})\\
\textit{Selection:} \\
\hspace*{2em} a) Priority Queue 0: Node 1 is extracted\\  
\hspace*{2em} b) Priority Queue 1: Node 6 is extracted\\
\textit{Relaxation: (Update Distances for Nodes)} \\  
\hspace*{2em} a) Priority Queue 0: 1) Node 5 : Min(Inf, 0 + 4)  2) Node 3 : Min(Inf, 0 + 7)  3) Node 0 : Min(Inf, 0 + 3)\\
\hspace*{2em} b) Priority Queue 1: 1) Node 4: Min(Inf, 0 + 4)  2) Node 2: Min(Inf, 0 + 6)  3)Node 7: Min(Inf, 0 + 5)\\

\textbf{Iteration 2:}(see Figure~\ref{fig:figure_2})\\
\textit{Selection:} \\
\hspace*{2em} a) Priority Queue 0: Node 0 is extracted\\
\hspace*{2em} b) Priority Queue 1: Node 4 is extracted\\
\textit{Relaxation: (Update Distances for Nodes)} \\
\hspace*{2em} a) Priority Queue 0: 1) Node 4: Min(Inf, 3 + 1) 2) Node 2: Min(Inf, 3 + 5)\\
\hspace*{2em} b) Priority Queue 1: 1) Node 5: Min(Inf, 4 + 8) 2) Node 0: Min(Inf, 4 + 1)\\

\textbf{Iteration 3:} (see Figure~\ref{fig:figure_3})\\
\textit{Selection:} \\
\hspace*{2em} a) Priority Queue 0: Node 5 is extracted\\
\hspace*{2em} b) Priority Queue 1: Node 7 is extracted\\
\textit{Relaxation: (Update Distances for Nodes)} \\
\hspace*{2em} a) Priority Queue 0: 1) Node 7: Min(Inf, 4 + 5) 2) Node 4: Min(4, 4 + 8)\\
\hspace*{2em} b) Priority Queue 1: 1) Node 3: Min(Inf, 5 + 3) 2) Node 5: Min(12, 5 + 2)\\

\textbf{Iteration 4:}(see Figure~\ref{fig:figure_4})\\
\textit{Selection:} \\
\hspace*{2em} a) Priority Queue 0: Node 4 is extracted, maxDist = 4\\
\hspace*{2em} b) Priority Queue 1: Node 0 is extracted, maxDist = 5\\
\textit{Relaxation: (Update Distances for Nodes)} \\
\hspace*{2em} a) Priority Queue 0: 1) Node 6: Min(Inf, 4 + 4)\\
\hspace*{2em} b) Priority Queue 1: 1) Node 1: Min(Inf, 5 + 3)\\

\textbf{Iteration 5:}(see Figure~\ref{fig:figure_5})\\
\textit{Selection:} \\
\hspace*{2em} a) Priority Queue 0: Node 7 is extracted, maxDist = 6\\
\hspace*{2em} b) Priority Queue 1: Node 2 is extracted\\
\textit{Relaxation: (Update Distances for Nodes)} \\
\hspace*{2em} a) Priority Queue 0: 1) Node 3: Min(7, 6 + 3) 2) Node 6: Min(8, 6 + 5)\\
\hspace*{2em} b) Priority Queue 1: 1) Node 1: Min(8, 6 + 2)\\

\textbf{Iteration 6:} (see Figure~\ref{fig:figure_6})\\
\textit{Selection:} \\
\hspace*{2em} a) Priority Queue 0: Node 3 is extracted\\
\hspace*{2em} b) Priority Queue 1: Node 5 is extracted, maxDist = 7\\
\textit{Relaxation: (Update Distances for Nodes)} \\
\hspace*{2em} a) Priority Queue 0: 1) Node 2: Min(8, 7+ 2)\\
\hspace*{2em} b) Priority Queue 1: 1) Node 1: Min(8, 7 + 4)\\

\textbf{Iteration 7:} (see Figure~\ref{fig:figure_7})\\
\textit{Selection:} \\
\hspace*{2em} a) Priority Queue 0: Node 2 is extracted, maxDist = 8\\
\hspace*{2em} b) Priority Queue 1: Node 3 is extracted, maxDist = 8\\
\textit{Relaxation: (Update Distances for Nodes)} \\
\hspace*{2em} a) Priority Queue 0: 1) Node 6: Min(8, 8 + 6)\\
\hspace*{2em} b) Priority Queue 1: 1) Node 3: Min(8, 8 + 7)\\

\textbf{Iteration 8:} (see Figure~\ref{fig:figure_8})\\
\textit{Selection:} \\
\hspace*{2em} a) Priority Queue 0: Node 6 is extracted, maxDist = 8\\
\hspace*{2em} b) Priority Queue 1: Node 1 is extracted, maxDist = 8\\

\textbf{Stopping}:\\
All nodes have been visited. Find node with minimum eccentricity (Node 4) and return node and corresponding maximum distances (maxDist = 4).

\end{example}

From the example, we can see that the smallest maximum distance occurs at Node 4 and therefore Node 4 is the center solution with a maximum distance of 4 for this example. In a future section, we will show that the stopping condition can be improved to significantly reduce the number of nodes processed while still being able to find the optimal center.

\subsection{Centroid}
\label{sec:Centroid}
In some scenarios, one may wish to find a solution that minimizes total time/distance to a common location instead of minimizing the maximum time/distance to a common location. This is equivalent to finding the centroid of the S-source-center problem. To find the centroid of the S-source-center problem we have to solve the following optimization problem:
\begin{equation}
\hat{v} = \argmin_{v_i \in V} \sum_{j}^{} d(v_i, s_j)
\end{equation}
where V is the set of all vertices within the graph and $s_j \in S$ represents the jth source node with S representing the set of all sources.

Similar to finding the center of the S-source-center problem, we can apply an algorithm based on a multiple source Dijkstra's Algorithm. Instead of optimizing for the function $f(v_i)$ we instead optimize for $g(v_i)$ where $g(v_i) = \sum_{j}^{} d(v_i, s_j)$ . The node with the minimum value of $g(v_i)$ is the centroid. The algorithm follows the same steps as section \ref{sec:Center}, except we keep track of the sum of distances to each source rather than the maximum distance to each source node.

\section{Algorithmic Optimization}

\subsection{Early Termination for Center}
\label{sec:Early Termination for Center}
By combining our algorithm with a better stopping condition, we can reduce the amount of nodes explored, increasing efficiency without sacrificing accuracy. In this section, we propose a stopping condition that adds an additional check upon finding the first intersection point . After finding the first intersection point $v_u$ (i.e. an extracted node that has been visited by all sources), a variable called \textit{minimax} is assigned the value of f($v_u$). The algorithm then continues the alternating Dijkstra's Algorithm but with an additional check whenever another intersection point is encountered. If the next intersection node discovered has a maximum distance value that is less than \textit{minimax}, then \textit{minimax} is updated to equal the node's maximum distance. Furthermore, we claim that a Dijkstra's Algorithm for a given source may be terminated if an extracted node has a distance larger than \textit{minimax}. The following Theorem makes this claim formal.

\begin{theorem*}
Assume that there are $S$ Dijkstra's Algorithms initiated from $S$ different source nodes. Once a node has been visited by all $S$ Dijkstra's Algorithms, then let these nodes be called $intersection$ nodes and let $d_{max}$ represent the maximum distance to all source nodes from one of the intersection nodes.  Then, for any of the $S$ Dijkstra's Algorithms, if a node is extracted from its priority queue with a distance that exceeds $d_{max}$, then all remaining nodes that are to be explored for that Dijkstra's Algorithm cannot result in a smaller maximum distance than $d_{max}$.
\end{theorem*}

\begin{proof}
Dijkstra's Algorithm extracts nodes in a non decreasing manner, implying that every extracted node's distance cannot be smaller than a previously extracted node's distance. If an extracted node has a distance, $d$, larger than $d_{max}$, then any extracted node after that will have a distance larger than $d$ and therefore larger than $d_{max}$. Therefore, any maximum distance that is calculated based on remaining nodes cannot result in a smaller maximum distance than $d_{max}$.
\end{proof}

Given the above Theorem, we know that if an extracted node has a distance greater than \textit{minimax}, then the extracted node and all successive nodes will result in a maximum distance larger than \textit{minimax} and therefore none of these nodes can be a center. Our new stopping condition for each Dijkstra's Algorithm is to check if an extracted node has distance larger than \textit{minimax} and if it does, then terminate that Dijkstra's Algorithm. Our new algorithm with a modified stopping condition is as follows (see Algorithm~\ref{alg:centerStoppingAlgorithm}):

%Algorithm 3
\begin{algorithm}
\caption{Multiple Source Dijkstra's Algorithm for Finding the Center with an Improved Stopping Condition}
\label{alg:centerStoppingAlgorithm} 
1) Initialization: Create a priority queue for each source node and initialize the distance of source nodes to 0 and all the other nodes to $\infty$. Set $minimax$ to $\infty$.

2) Selection: Pick the unvisited node with the smallest d($s_i$,v) (initially the source node) and extract node

3) Relaxation: For the extracted node, check all neighboring nodes and compare/update its distances. If extracted node has been visited by all sources, compare its maximum distance to $minimax$.  If it is less than $minimax$, then set $minimax$ to this maximum distance.

4) Alternation: The extracted node is marked as visited and will not be visited again. If the extracted node has a distance that is larger than $minimax$ then go to step 5.  Otherwise, repeat steps 2-3 alternating between each of the sources that have not been terminated, for all nodes.

5) Stopping: If all sources have been terminated then the center is the node with value $minimax$ and $minimax$ is the smallest maximum distance.  Otherwise, repeat steps 2-3 alternating between each of the sources that have not been terminated, for all nodes. \\
\end{algorithm}

To illustrate how the algorithm works, we provide the following example. In the  example (see Example~\ref{ex:example2}) we have two source nodes, Nodes 1 and 6. The graph is both undirected and contains all positive, but not identical weights and has 8 total vertices and 12 total edges.

%Example 2
\begin{example}[label={ex:example2}]{Finding the Center with A Stopping Condition}
\textbf{Initialization:}\\
\hspace*{2em} Create a priority queue for both sources\\
\hspace*{2em} Set \textit{minimax} to $\infty$\\
\hspace*{2em} b) Priority Queue 0: Node 1 distance is set to 0 and all others are set to $\infty$\\
\hspace*{2em} a) Priority Queue 1: Node 6 distance is set to 0 and all others are set to $\infty$\\

Iterations 1-3 are the same as Example~\ref{ex:example1} and Iteration 4 contains a stopping condition.\\

\textbf{Iteration 4:} (see Figure~\ref{fig:figure_4})\\
\textit{Selection:} \\
\hspace*{2em} a) Priority Queue 0: Node 4 is extracted\\
\hspace*{2em} b) Priority Queue 1: Node 0 is extracted\\
\textit{Relaxation: (Update Distances for Nodes)} \\
\hspace*{2em} Compare maxDist of Node 0 and maxDist of Node 4 to \textit{minimax} and set \textit{minimax} to 4\\
\hspace*{2em} a) Priority Queue 0: 1) Node 6: Min(Inf, 4 + 4)\\
\hspace*{2em} b) Priority Queue 1: 1) Node 1: Min(Inf, 5 + 3)\\
\textbf{Stopping}:\\
The maxDist value for Node 0 is 5 which is greater than the \textit{minimax}, therefore terminate Dijkstra's Algorithm from Source Node 6.\\

\textbf{Iteration 5:} (see Figure~\ref{fig:figure_5})\\
\textit{Selection:}\\
\hspace*{2em} a) Priority Queue 0: Node 7 is removed\\
\textbf{Stopping:}\\
The maxDist value for Node 7 is 6 which is greater than the \textit{minimax}, ending the Dijkstra's Algorithm from Source Node 1.  Since all Dijkstra's Algorithms are terminated the algorithm is also terminated with minimax = 4 at node 4.
\end{example}

From the example, we can see that the smallest maximum distance (i.e., \textit{minimax}) is equal to 4 at the end of the algorithm and node 4 is where \textit{minimax} occurs.  So, the center solution is node 4 with a maximum distance of 4 for this example. The stopping condition has greatly reduced the number of nodes explored as both Source Nodes 1 and 6 only explored 4 out of the 8 nodes. This example explored less nodes than example~\ref{ex:example1}, resulting in only 50\% of nodes being explored.

\subsection{Early Termination for Centroid}
\label{sec:Early Termination for Centroid}

The theorem that led to the stopping condition for finding the center does not apply for finding the centroid because the stopping condition utilizes the fact that the priority queue sorts nodes based on their minimum distance value. However, the priority queue for a given Dijkstra's Algorithm is not necessarily aware of the distance from a node to other sources and therefore can not explore nodes in the order of increasing sum of distances. In particular, the theorem that we used to create the stopping condition for finding the center does not directly apply for finding the centroid. We can still however use the theorem as a heuristic to find a relative minimum of the sum of distances to the sources. Specifically, we can check for any intersection node whether the sum of distances is larger than the previously stored minimum sum of distances. If it is, then we can update this minimum sum of distances. Otherwise, we can terminate the Dijkstra's Algorithm that resulted in a larger sum of distances. Our new algorithm with a modified stopping condition is as described in  Algorithm~\ref{alg:centroidStoppingAlgorithm}. To illustrate how the algorithm works, we provide an example (see Example~\ref{ex:example3}).

%Algorithm 4
\begin{algorithm}
\caption{Multi Sourced Dijkstra's Algorithm for Finding Centroid}
\label{alg:centroidStoppingAlgorithm}
1) Initialization: Create a priority queue for each source node and initialize the distance of source nodes to 0 and all the other nodes to $\infty$. Set $minsum$ to $\infty$.

2) Selection: Pick the unvisited node with the smallest d($s_i$,v) (initially the source node) and extract node

3) Relaxation: For the extracted node, check all neighboring nodes and compare/update its distances. If extracted node has been visited by all sources, compare its sum to $minsum$.  If it is less than $minsum$, then set $minsum$ to this sum.

4) Alternation: The extracted node is marked as visited and will not be visited again. If the extracted node has been visited by all sources and results in a sum that is larger than $minsum$ then go to step 5.  Otherwise, repeat steps 2-3 alternating between each of the sources that have not been terminated, for all nodes.

5) Stopping: If all sources have been terminated then the centroid is the node with value $minsum$ and $minsum$ is the smallest sum. Otherwise, repeat steps 2-3 alternating between each of the sources that have not been terminated, for all nodes. \\
\end{algorithm}

%Example 3
\begin{example}[label={ex:example3}]{Finding the Centroid with A Stopping Condition}
\textbf{Initialization:}\\
\hspace*{2em} Create a priority queue for both sources\\
\hspace*{2em} Set minsum to $\infty$\\
\hspace*{2em} a) Priority Queue 0: Node 1 distance is set to 0 and all others are set to $\infty$\\
\hspace*{2em} b) Priority Queue 1: Node 6 distance is set to 0 and all others are set to $\infty$\\

Iterations 1-3 are the same as Example~\ref{ex:example1} and Iteration 4 contains a stopping condition.\\

\textbf{Iteration 4:} (see Figure~\ref{fig:figure_4} )\\
\textit{Selection:} \\
\hspace*{2em} a) Priority Queue 0: Node 4 is extracted\\
\hspace*{2em} b) Priority Queue 1: Node 0 is extracted\\
\textit{Relaxation: (Update Distances for Nodes)} \\
\hspace*{2em} a) Priority Queue 0: 1) Node 6: Min(Inf, 4 + 4)\\
\hspace*{2em} b) Priority Queue 1: 1) Node 1: Min(Inf, 5 + 3)\\
\textbf{Stopping}:\\
First Intersection Node is found at Node 4 and Node 0\\
The variable $minsum$ is assigned the value of the sum of distances from Node 4 to each of the sources which is 8. The sum of distances from Node 0 to each of the sources is also 8 thus resulting in the continuation of both Dijkstra's Algorithms.\\

\textbf{Iteration 5:} (see Figure~\ref{fig:figure_5} )\\
\textit{Selection:} \\
\hspace*{2em} a) Priority Queue 0: Node 7 is extracted\\
\hspace*{2em} b) Priority Queue 1: Node 2 is extracted\\
\textit{Relaxation: (Update Distances for Nodes)} \\
\hspace*{2em} a) Priority Queue 0: 1) Node 3: Min(7, 6 + 3) 2) Node 6: Min(8, 6 + 5)\\
\hspace*{2em} b) Priority Queue 1: 1) Node 1: Min(8, 6 + 2)\\
\textbf{Stopping}:\\
The next intersection node is Node 7. The value of the extracted Node 7 has a sum of distances equal to 11, which is greater than the $minsum$ value of 8, terminating the Dijkstra's Algorithm for Source Node 1.\\

\textbf{Iteration 6:} (see Figure~\ref{fig:figure_6} )\\
\textit{Selection:} \\
\hspace*{2em} b) Priority Queue 1: Node 5 is extracted\\
\textit{Relaxation: (Update Distances for Nodes)} \\
\hspace*{2em} b) Priority Queue 1: 1) Node 1: Min(8, 7 + 4)\\
\textbf{Stopping}:\\
The next intersection node is Node 5. The value of the extracted Node 5 has a sum of distances equal to 11, which is greater than the $minsum$ value of 8, terminating the Dijkstra's Algorithm for Source Node 6 and for the whole algorithm. The resulting value of $minsum$ is 8, which occurs at node 4.
\end{example}

From the example, we can see that the smallest sum occurs at node 4 and therefore node 4 is the centroid solution with a sum of 8 for this example. The percent of nodes that are explored is 62.5\% for Source Node 1 and 75\% for Source Node 6 resulting in a total of 68.75\% of nodes being explored.  In this example the solution that was found happens to coincide with the optimal centroid solution, but this is not always the case.  In Figure~\ref{fig:figure_9} we provide an example in which we do not get an optimal solution by utilizing algorithm~\ref{alg:centroidStoppingAlgorithm}. This happens because algorithm~\ref{alg:centroidStoppingAlgorithm} ends up finding a relative minimum instead of a global minimum. In the Figure, the algorithm terminates at node 3 with a sum distance of 13, however, there exists a smaller sum of 11 at node 5.

\subsection{Solution Utilizing A* Algorithm}
\label{sec:Solution Utilizing A*}
To adapt A* algorithm to our problem of finding the center of $S$ sources, we can utilize the same algorithm as in algorithm~\ref{alg:centerAlgorithm} except we will run A* algorithm from each source node. Since A* algorithm without a heuristic function is essentially the same as Dijkstra's Algorithm, applying a heuristic function to A* algorithm will reduce the number of explored nodes without reducing accuracy. The heuristic function that we will use guides the A* algorithm based off of the distances to each source node. We introduce the cost function as:

\begin{equation}
f(n) = g(n) + h(n)
\end{equation}
where the heuristic function is:
\begin{equation}
h(n) = \max_{s_j \in S}(d(v_{i}, s_{j}))
\end{equation}
and $v_{i}$ is the current vertex and $S$ is the set of all sources.

The heuristic function is the maximum distance from the current vertex to all sources. The same heuristic function can also be used for finding the centroid. Note that the heuristic function should be chosen such that it is smaller than the actual distance between two nodes. In the case that distance refers to time, we may choose to use a heuristic function that is equal to the maximum distance divided by the maximum speed to all sources.

A* algorithm maintains the same priority queue as Dijkstra's Algorithm, so we are able to alternate between sources to have an alternating A* algorithm. Note that the stopping condition derived from our theorem may not be optimal when applied to A* algorithm. This is because A* algorithm does not necessarily process nodes in the order of increasing $g(n)$, but rather processes nodes in the order of increasing $f(n)$. Nonetheless, one may choose to use our theorem as a heuristic in combination with A* algorithm to further reduce compute at the expense of sacrificing some accuracy.

\section{Experimental Setup}
\label{sec:Experimental Setup}
To test the algorithms described in the previous sections, we wrote C code to test the percent of nodes explored by using the optimized algorithm for finding the center algorithm as well as the accuracy of the percent of nodes explored by using the optimized algorithm for finding the centroid (see section \ref{sec:Centroid}). We also measure the accuracy of our algorithm for finding the centroid node. We wrote C code to randomly generate graphs. For a fixed number of vertices and source nodes, our program would randomly generate directed edges with random weights between 1 and 100. The program would also randomly choose $n$ vertices within the graph as source nodes. We discounted any graphs where the graph was disconnected and were unable to reach any intersecting nodes. In practice, most of the time, it will be a small number of people who want to find a meeting place, but in some cases, there may be more people who want to find a meeting place. As a result, we chose to test {2, 3, 5, 10} sources. Additionally, we utilized a wide range of vertices ({20, 50, 100, 500}) to try and capture both small and large graphs. By running 1000 iterations for each of the combinations between source nodes {2, 3, 5, 10} and vertices {20, 50, 100, 500}, we were able to obtain a wide range of results. Specifically, the randomly chosen nodes generated graphs that are representative of both sparse and dense graphs.

Note that we did not run simulations using A* algorithm, as the savings achieved from A* algorithm is orthogonal to the savings achieved from the stopping conditions that we proposed. One may choose to combine A* algorithm with our stopping conditions to achieve further savings, but this is left as a future work.

\section{Experimental Results}
\label{sec:Experimental Results}

\begin{table}[h!]
\begin{center}
\begin{tabular}{cccccl}\toprule
& \multicolumn{4}{c}{Number of Vertices} 
\\\cmidrule(lr){2-5}
        Number of \\Source Nodes  &  $20$  & 50 & 100 & 500  \\\midrule
2  & $28.483870$ & $18.643416$ & $13.993495$  & $7.948100$ & \\
    & $14.61$ & $10.74$ & $9.50$ & $4.92$ & \\
3 & $42.058662$ & $30.736375$ & $25.613805$ & $18.473196$ \\
   & $14.81$ & $11.27$ & $11.09$ & $8.60$ & \\
5 & $56.208606$ & $45.070680$ & $40.851895$ & $35.130090$ \\
   & $13.51$ & $12.43$ & $11.68$ & $10.98$ & \\
10  & $70.214316$ & $61.115701$ & $57.770726$ & $55.757189$ \\
      & $11.74$ & $11.35$ & $11.11$ & $12.41$ & \\\bottomrule
\end{tabular}
\caption{Average (top) and standard deviation (bottom) of percentage of nodes explored for minimum of maximum distances}
\label{tab:table1}
\end{center}
\end{table}

In our first experiment, we tried finding the center of various graphs using algorithm~\ref{alg:centerStoppingAlgorithm} and compared it to algorithm~\ref{alg:centerAlgorithm}. We compared the percentage of nodes explored using algorithm~\ref{alg:centerStoppingAlgorithm} versus the number of nodes explored in algorithm~\ref{alg:centerAlgorithm} (see Table~\ref{tab:table1}). Mathematically, the percentage is calculated by the following equation. 
\begin{equation}
percentage = \frac{k} {n} *100
\end{equation}

\noindent where $k$ represents the number of nodes explored with algorithm~\ref{alg:centerStoppingAlgorithm} and $n$ represents the number of nodes explored with algorithm~\ref{alg:centerAlgorithm}. In Table~\ref{tab:table1}, we see that the average percentage of nodes explored is directly proportional to the number of source nodes. The reason for this is because as the number of source nodes increases, the likelihood of one of the source nodes being further away from the center increases, hence resulting in more nodes being explored before an intersection node is even found. On the other hand, the average percentage of nodes explored is inversely proportional to the number of vertices due to our stopping condition. Specifically, our stopping condition allows for efficient termination once an intersection node is found and a larger number of vertices implies a larger number of edges which increases the likelihood of shorter paths to intersection nodes. As a result, fewer nodes need to be explored for a graph that has more vertices (e.g. 500) with random edges than a graph with less vertices (e.g. 20). We notice anywhere from a 1.5x to a 12x savings in nodes explored depending on variations in the amount of vertices and number of nodes for the simulation.

\begin{table}[h!]
\begin{center}
\begin{tabular}{cccccl}\toprule
& \multicolumn{4}{c}{Number of Vertices} 
\\\cmidrule(lr){2-5}
        Number of \\Source Nodes  &  $20$  & 50 & 100 & 500  \\\midrule
2  & $35.342591$ & $21.790436$ & $15.785197$  & $7.319593$ & \\
    & $15.54$ & $10.16$ & $9.84$ & $2.40$ & \\
3 & $49.674350$ & $33.779958$ & $26.540579$ & $16.002276$ \\
   & $13.15$ & $8.07$ & $8.05$ & $3.55$ & \\
5 & $65.202754$ & $48.324506$ & $40.517998$ & $29.499630$ \\
   & $9.80$ & $8.31$ & $7.77$ & $4.25$ & \\
10  & $81.099321$ & $66.013424$ & $57.087868$ & $47.448422$ \\
      & $6.68$ & $6.48$ & $6.25$ & $5.57$ & \\\bottomrule
\end{tabular}
\caption{Average (top) and standard deviation (bottom) of percentage of nodes explored for minimum of sum distances}
\label{tab:table2}
\end{center}
\end{table}

The next experiment that we ran was to find the centroid of various graphs using algorithm~\ref{alg:centroidStoppingAlgorithm} and compared it to algorithm~\ref{alg:centerAlgorithm}. We compared the percentage of nodes explored using algorithm~\ref{alg:centroidStoppingAlgorithm} to the number of nodes explored using algorithm~\ref{alg:centerAlgorithm} (see Table~\ref{tab:table2}). Mathematically, the percentage is calculated as:  
\begin{equation}
percentage = \frac{k} {n} *100
\end{equation}
where $k$ represents number of nodes explored with algorithm~\ref{alg:centroidStoppingAlgorithm} and $n$ represents number of nodes explored with algorithm~\ref{alg:centerAlgorithm}. In Table~\ref{tab:table2}, like Table~\ref{tab:table1}, we can draw many of the same conclusions as the previous experiment because  algorithm~\ref{alg:centroidStoppingAlgorithm} maintains many of the same aspects as algorithm~\ref{alg:centerStoppingAlgorithm}. However, we do see that the averages in Table~\ref{tab:table2} are almost always larger than the ones in Table~\ref{tab:table1}. This is due to the fact that in algorithm~\ref{alg:centroidStoppingAlgorithm}, nodes must be marked as visited by all sources before comparing a sum for termination whereas in algorithm~\ref{alg:centerStoppingAlgorithm}, a node does not need to be visited by every source for checking the termination condition. This stricter termination condition results in more nodes being explored in the case of the centroid algorithm. Similar to the previous experiment, we notice a savings in nodes explored of around 2x to 12x depending on variations in the amount of vertices and number of nodes for the simulation.

\begin{table}[h!]
\begin{center}
\begin{tabular}{cccccl}\toprule
& \multicolumn{4}{c}{Number of Vertices} 
\\\cmidrule(lr){2-5}
        Number of \\Source Nodes  &  $20$  & 50 & 100 & 500  \\\midrule
2  & $108.514600$ & $112.411543$ & $113.197901$  & $111.620521$ & \\
    & $15.54$ & $10.16$ & $9.84$ & $2.40$ & \\
3 & $105.778272$ & $107.962289$ & $108.237160$ & $108.429273$ \\
   & $13.15$ & $8.07$ & $8.05$ & $3.55$ & \\
5 & $102.676250$ & $103.993519$ & $104.353387$ & $106.003617$ \\
   & $9.80$ & $8.31$ & $7.77$ & $4.25$ & \\
10  & $100.654489$ & $101.133570$ & $101.613492$ & $103.108884$ \\
      & $6.68$ & $6.48$ & $6.25$ & $5.57$ & \\\bottomrule
\end{tabular}
\caption{Average (top) and standard deviation (bottom) of the percent difference between the sum that results from the proposed algorithm over the sum that results from the optimal algorithm}
\label{tab:table3}
\end{center}
\end{table}

Due to the fact that the solution found in algorithm~\ref{alg:centroidStoppingAlgorithm} may not be the optimal centroid, we also want to measure how far (on average) the solution is from the optimal. In Table~\ref{tab:table3}, we calculate the percent difference of the sum that results from algorithm~\ref{alg:centroidStoppingAlgorithm} over the sum that results from the algorithm described in section \ref{sec:Centroid}. As the number of source nodes increases, the accuracy approaches a near perfect value of 100\%. While the accuracy is not perfect for various simulation parameters, it is more than good enough for practical situations given its savings of nodes explored of nearly 2x to 12x for most cases. Recall that our end goal is to find a "fair" meeting place for S people, so if the resulting centroid solution is a few percent larger than the optimal centroid solution, this will most likely be acceptable.

\begin{table}[h!]
\begin{center}
\begin{tabular}{cccccl}\toprule
& \multicolumn{4}{c}{Number of Vertices} 
\\\cmidrule(lr){2-5}
        Number of \\Source Nodes  &  $20$  & 50 & 100 & 500  \\\midrule
2 & $77.8367$ & $68.961973$ & $63.333333$  & $63.947633$ & \\
3 & $78.3405$ & $67.598344$ & $61.663286$ & $55.790534$ \\
5 & $84.0088$ & $74.948025$ & $65.853659$ & $48.036254$ \\
10  & $91.6667$ & $83.874346$ & $77.323800$ & $52.366566$ \\\bottomrule
\end{tabular}
\caption{Average of Accuracy for Centroid Locating Algorithm}
\label{tab:table4}
\end{center}
\end{table}

Finally, we also measure the percentage of times that algorithm~\ref{alg:centroidStoppingAlgorithm} finds the true centroid (see table~\ref{tab:table4}). We notice that the accuracy decreases as the number of vertices increases and increases as the number of source nodes increases. This is consistent with our previous observations. As the number of source nodes increases, the number of nodes explored also increases, resulting in a higher likelihood of finding the true centroid. Similarly, as the number of vertices increases, the number of nodes explored decreases resulting in a lower likelihood of finding the true centroid.

\section{Conclusion}
\label{sec:Conclusion}
We proposed solutions to the problem of finding the center of S sources by utilizing S different Dijkstra's Algorithms for both the centroid and center. Furthermore, we proposed optimizations to reduce the amount of exploration of said algorithm by a factor of 2x to 12x. We also provided an optimal solution for finding the center of a graph utilizing S Dijkstra's Algorithms and a stopping condition that significantly reduces the time complexity. While we were unable to find an optimal solution for the centroid, we found an algorithm with less than 10 percent degradation in accuracy on average, but still maintained savings of nodes explored of 2x to 12x. We addressed the real world problem of finding a "fair" meeting place through experimentation and a usage of a wide variety of randomly generated graphs. The problem is applicable for mapping applications and our optimizations allow us to find a meeting spot for S different people with efficient time complexity.

For future work, we hope to find an algorithm that can lead to efficient updates for the center/centroid when a small number of edge weights change as a function of time. Other directions for future work are to explore combining the A* algorithm with our stopping condition and also to improve the accuracy of our solution to finding the centroid.

\bibliographystyle{plainnat}
\bibliography{NSourceBib}

\end{document}